\documentclass[conference]{IEEEtran}

\IEEEoverridecommandlockouts
\usepackage{amsmath,amsthm,amsfonts}
\usepackage{graphicx}
\usepackage{epsf}
\usepackage{mdwmath}
\usepackage{mdwtab}
\usepackage{cite}
\usepackage{amssymb}
\usepackage{algorithmic}
\usepackage{algorithm}

\usepackage{amssymb}
\usepackage{algorithmic}
\usepackage{algorithm}
\usepackage{bm}

\newtheorem{thm}{Theorem}[section]

\newtheorem{lem}[thm]{Lemma}
\newtheorem{definition}[thm]{Definition}

\floatname{algorithm}{Sequential Thresholding \hspace{50cm}}

\newcommand{\E}{{\cal{E}}}

\def\E{\mathbb{E}}
\def\P{\mathbb{P}}
\def\1{\bm{1}}

\def\S{\mathcal{S}}

\begin{document}
\pagestyle{empty}

\title{On the Limits of Sequential Testing in High Dimensions}

\author{\IEEEauthorblockN{Matt Malloy}
\IEEEauthorblockA{ Electrical and Computer Engineering\\
University of Wisconsin-Madison\\
Email: mmalloy@wisc.edu}
\and
\IEEEauthorblockN{Robert Nowak}
\IEEEauthorblockA{
Electrical and Computer Engineering\\
University of Wisconsin-Madison\\
Email: nowak@ece.wisc.edu}}
\maketitle
\vspace{-2cm}


\begin{abstract}
This paper presents results pertaining to sequential methods for support recovery of sparse signals in noise.  Specifically, we show that \emph{any} sequential measurement procedure fails provided the average number of measurements per dimension grows slower then $D(f_0||f_1)^{-1} \log s$ where $s$ is the level of sparsity, and $D(f_0||f_1)$ the Kullback–-Leibler divergence between the underlying distributions.  For comparison, we show any \emph{non-sequential} procedure fails provided the number of measurements grows at a rate less than $D(f_1||f_0)^{-1} \log n $, where $n$ is the total dimension of the problem.  Lastly, we show that a simple procedure termed \emph{sequential thresholding} guarantees exact support recovery provided the average number of measurements per dimension grows faster than $D(f_0||f_1)^{-1} ( \log s + \log \log n)$, a mere additive factor more than the lower bound.
\end{abstract}

\section{Introduction}
High-dimensional signal support recovery is a fundamental problem arising in many aspects of science and engineering.  The goal of the basic problem is to determine, based on noisy observations, a sparse set of elements that somehow differ from the others.

In this paper we study the following problem.  Consider a support set $\S \subset \{1,...,n\}$ and a random variable $y_{i,j}$ such that
\begin{eqnarray} \label{eqn:underlyingstats}
    y_{i,j} \sim \left\{
           \begin{array}{ll}
             f_0(\cdot) & i \not \in \mathcal{S} \\
             f_1(\cdot) & i \in \mathcal{S}
           \end{array}
         \right.
\end{eqnarray}
where $f_0(\cdot)$ and $f_1(\cdot)$ are probability measures on $\mathcal{Y}$, and $j$ indexes multiple independent measurements of any component $i \in \{1,...,n\}$.  The dimension of the problem, $n$, is large --  perhaps thousands or millions or more -- but the support set $\S$ is sparse in the sense that the number of elements following $f_1$ is much less than the dimension, i.e., $|\S| = s \ll n$. The goal of the sparse recovery problem is to identify the set $\S$.

In a \emph{non-sequential} setting $m\geq 1$ independent observations of each component are made ($y_{i,1},...,y_{i,m}$ are observed for each $i$) and the fundamental limits of reliable recovery are readily characterized in terms of Kullback-Leibler divergence and dimension.

Sequential approaches to the high dimensional support recovery problem have been given much attention recently (see \cite{Haupt}, \cite{MalloyISIT}, \cite{5688486}, \cite{BarronISIT}, etc). In the \emph{sequential} setting, the decision to observe $y_{i,j}$ is based on prior observations, i.e. $y_{i,1},...,y_{i,j-1}$.  Herein lies the advantage of sequential methods: if prior measurements indicate a particular component belongs (or doesn't belong) to $\S$ with sufficient certainty, measurement of that component can cease, and resources can be diverted to a more uncertain element.

The results presented in this paper are in terms of asymptotic rate at which the average number of measurements per dimension, denoted $m$, must increase with $n$ to ensure exact recovery of $\S$ for any fixed distributions $f_0$ and $f_1$.  The main contributions are \emph{1)} to present a necessary condition for success of any sequential procedure in the sparse setting, \emph{2)} show success of a simple sequential procedure first presented in \cite{MalloyISIT} is guaranteed provided the average number of measurements per dimension is within a small additive factor of the necessary condition for any sequential procedure, and compare this procedure to the known optimal sequential probability ratio test, and \emph{3)} lastly, compare these results to the performance limits of any non-sequential procedure.  Table \ref{table1} summarizes these results.  

\begin{table}[ht!]
\caption{Average number of measurements per dimension for exact recovery} 
\centering
\begin{tabular}{p{2.8cm}|p{2.2cm}|p{2.3cm}}
 \hline
  \hline
  \emph{Non-sequential}            &  $m \geq \frac{\log n }{D(f_1||f_0)}$ & necessary \vspace{.2cm} \\
  \hline
  \emph{Sequential}  &  $m \geq \frac{\log s }{D(f_0||f_1)}$ & necessary \vspace{.2cm}\\
  \hline
  \emph{Sequential Thresholding}   &  $m > \frac{\log s }{D(f_0||f_1)} + \frac{\log \log n }{D(f_0||f_1)}$ & sufficient if $s$ \newline sub-linear in $n$ \\
  \hline
\end{tabular}
\label{table1}
\end{table}

Our results are striking primarily for two reasons.  First, \emph{sequential} procedures succeed when the number of measurements per dimension increases at a rate logarithmic in the level of \emph{sparsity}, i.e. $\log s$.  In contrast, \emph{non-sequential} procedures require the average number of measurements per dimension to increase at a rate logarithmic in the \emph{dimension}, i.e. $\log n$.  For signals where sparsity is sublinear in dimension, the gains of sequential methods are polynomial; in scenarios where the sparsity grows logarithmically, the gains are exponential.

Secondly, and perhaps equally as surprising, a simple procedure dubbed sequential thresholding achieves nearly optimal performance provided minor constraints on the level of sparsity are met (specifically, that $s$ is sublinear in $n$). In terms of the average number of measurements per dimension, the procedure comes within an additive factor, doubly logarithmic in dimension, of the lower bound of any sequential procedure.

\section{Problem Formulation}
Let $\S$ be a sparse subset of $\{1,...,n\}$ with cardinality $s = |\S|$.   For any index $i \in \{1,...,n\}$, the random variable $y_{i,j}$ is independent, identically distributed according to (\ref{eqn:underlyingstats}).  That is, for all $j$, $y_{i,j}$ follows distribution $f_1(\cdot)$ if $i$ belongs to $\S$, and follows $f_0(\cdot)$ otherwise.  We refer to $f_0$ as the null distribution, and $f_1$ the alternative.

In this paper, we limit our analysis to exact recovery of $\S$ using coordinate wise methods.  Defining $\hat{\S}$ as an estimate of $\S$, the family wise error rate is given as:
\begin{eqnarray} \nonumber
  \P({\cal{E}}) = \P(\hat{\S} \neq \S) =  \P\left(\bigcup_{i \not \in \S} {\cal{E}}_{i} \cup \bigcup_{i \in \S } {\cal{E}}_{i} \right)
\end{eqnarray}
where ${\cal{E}}_{i}$, $i \not \in \S$ is a false positive error event and ${\cal{E}}_{i}$, $i \in \S$ a false negative error event.  To simplify notation, we define the false positive and false negative probabilities in the usual manner: $\alpha = \P({\cal{E}}_{i}| i \not \in \S)$, and $\beta =  \P({\cal{E}}_{i}|i \in \S)$.

The test to decide if component $i$ belongs to $\S$ is based on the normalized log-likelihood ratio.  For $y_j$ distributed i.i.d. $f_0$ or $f_1$,
\begin{eqnarray} \nonumber
  t^{(m)} := \frac{1}{m} \sum_{j=1}^m \log \frac{f_1(y_{j})}{f_0(y_{j})}
\end{eqnarray}
which is a function of $ (y_{1},...,y_{m}) \in \mathcal{Y}^m$.  The superscript $m$ explicitly indicates the number of measurements used to form the likelihood ratio and is suppressed when unambiguous. The log-likelihood ratio is compared against a scalar threshold $\gamma$ to hypothesize if a component follows $f_0$ or $f_1$:
\begin{eqnarray} \nonumber
    t   \gtrless^{f_1}_{f_0} \gamma.
\end{eqnarray}

Additionally, the Kullback-Liebler divergence of distribution $f_0$ from $f_1$ is defined as:
\begin{eqnarray} \nonumber
  D(f_1||f_0) = \E_1\left[ \log \frac{f_1(y)}{f_0(y)} \right]
\end{eqnarray}
where $\E_1\left[\cdot \right]$ is expectation with respect to distribution $f_1$.

\subsection{Measurement procedures}
To be precise in characterizing a measurement procedure, we continue with three definitions.
\begin{definition}
\emph{Measurement Procedure.}
A procedure, denoted $\pi$, used to determine if $y_{i,j}$ is measured.  If $\pi_{i,j} = 1$, then $y_{i,j}$ is measured.  Conversely, if $\pi_{i,j} = 0$, then $y_{i,j}$ is not measured.
\end{definition}
\begin{definition}
\emph{Non-sequential measurement procedure.}
Any measurement procedure $\pi$ such that $\pi_{i,j}$ is not a function of $y_{i,j'}$ for any $j'$.
\end{definition}
\begin{definition}
\emph{Sequential measurement procedure.}
A measurement procedure $\pi$ in which $\pi_{i,j}$ is allowed to depend on prior measurements, specifically, $\pi_{i,j}: \{y_{i,1},...,y_{i,j-1}\} \mapsto \{0,1\}$.
\end{definition}

\subsection{Measurement Budget}
In order to make fair comparison between measurement schemes, we limit the total number of observations of $y_{i,j}$ in expectation.  For any procedure $\pi$, we require
\begin{eqnarray} \label{eqn:budget}
 \mathbb{E}\left[\sum_{i,j} \pi_{i,j} \right]\leq nm
\end{eqnarray}
for some integer $m$.  This implies, on average, we use $m$ or fewer observations per dimension.

\section{Sequential Thresholding}
Sequential thresholding, first presented in \cite{MalloyISIT}, relies on a simple bisection idea.  The procedure consists of a series of $K$ measurement passes, where each pass eliminates from consideration a proportion of the components measured on the prior pass.  After the last pass the procedure terminates and the remaining components are taken as the estimate of $\S$.  Sequential thresholding is described in the following algorithm.

\begin{algorithm}[h]
\caption{}
\begin{algorithmic} \label{alg:sds}
\STATE{input: $K \approx \log n$ steps, threshold $\gamma$}
\STATE{initialize: ${\bm{\pi}}_{i,1} = \bf{1}$} for all $i$
\FOR{$k = 1,\dots,K$}
\FOR{$\{i : \bm{\pi}_{i,k} = 1\}$ }
\STATE{{\bf measure}: $t_i$}
\STATE{{\bf threshold}: ${\bm{\pi}}_{i,k+1}  = \left\{
                                                \begin{array}{ll}
                                                  \bm{1} & t_i > \gamma \\
                                                  \bm{0} & \mathrm{else}
                                                \end{array}
                                              \right.
$}
\ENDFOR
\ENDFOR
\STATE{output: $\hat{\S} = \{i:\bm{\pi}_{i,K+1} =\bm{1}\}$}
\end{algorithmic}
\end{algorithm}

\subsection{Example of Sequential Thresholding}
Sequential thresholding is perhaps best illustrated by example.  Consider a simple case with measurement budget $m=2$, $f_0 \sim {\cal{N}}(0,1)$ and $f_1 \sim {\cal{N}}(\theta,1)$ for some $\theta > 0$.

On the first pass, the procedures measures $y_{i,1}$ for all $i$, using $n$ measurements (half of the total budget as $mn = 2n$).  On subsequent passes, the procedure observes $y_{i,k}$ if $\pi_{i,k} = 1$.  To set $\pi_{i,k+1}$ the procedure thresholds observations that fall below, for example, $\gamma = 0$, eliminating a proportion (approximately half in this case) of components following the null distribution:
\begin{eqnarray} \nonumber
  \pi_{i,k+1} = \left\{
                \begin{array}{ll}
                  1 & y_{i,k} > 0 \\
                  0 & y_{i,k} \leq 0.
                \end{array}
              \right.
\end{eqnarray}
In words, if a measurement of component $i$ falls below the threshold on any pass, that component is \emph{not} measured for the remainder of the procedure, and not included in the estimate of $\S$. After $K \approx \log n $ passes, the procedure terminates, and estimates $\S$ as the set of indices that have not been eliminated from consideration:  $\hat{\S} = \{i:\pi_{i,K+1} =1\}$.

\subsection{Details of Sequential Thresholding}
Sequential thresholding requires two inputs: \emph{1)} $K$, the number of passes, and \emph{2)} $\gamma$, a threshold.  We define $\rho$ as the probability a component following the null is eliminated on any given pass, which is related to the threshold as
\begin{eqnarray} \nonumber
  \P(t_i^{(\rho m)} \leq \gamma | i \not \in \S) = \rho.
\end{eqnarray}
Additionally, we restrict our consideration to $\rho \in [1/2,1)$ -- that is, the probability a null component is eliminated on a given pass is one half or greater.

On each pass, $\rho m$ (which we assume to be an integer) measurements of a subset of components are made, and the log-likelihood ratio $t_i^{(\rho m)}$ is formed for each component.  
As measurements are made in blocks of size $\rho m$, we use \emph{boldface} $\bm{\pi}_{i,k}$ to indicate a block of measurements are taken of component $i$ on the $k$th measurement pass.  $\bm{\pi}_{i,k}$ can be interpreted as a vector:
\begin{eqnarray}  \nonumber
  \bm{\pi}_{i,k} = (\pi_{i,(k-1)\rho m+1},...,\pi_{i,\rho m}).
\end{eqnarray}

With $\gamma$ and $K \approx \log n$ as inputs, sequential thresholding operates as follows.  First, the procedure initializes, setting $\bm{\pi}_{i,1} = \bm{1}$.  For passes $k = 1,...,K$ the procedure measures $t_{i}^{(\rho m)}$ if $\bm{\pi}_{i,k} = \bm{1}$.  To set $\bm{\pi}_{i,k+1}$, the procedure tests the corresponding log-likelihood ratio against the threshold $\gamma$:
\begin{eqnarray} \nonumber
  \bm{\pi}_{i,k+1} &=& \left\{
                      \begin{array}{ll}
                        \bm{1} &\mathrm{if} \quad  t_i^{(\rho m)} > \gamma \\
                        \bm{0} &\mathrm{else}.
                      \end{array}
                    \right. \\ \nonumber
\end{eqnarray}
That is, if $t_i^{(\rho m)}$ is below $\gamma$, no further measurements of component $i$ are taken.  Otherwise, component $i$ is measured on the subsequent pass.  By definition of $\gamma$, approximately $\rho$ times the number of remaining components following $f_0$ will be eliminated on each pass; if $s \ll n$, each thresholding step eliminates approximately $\rho$ times the total number of components remaining.

After pass $K$, the procedure terminates and estimates $\S$ as the indices still under consideration: $\hat{\S} = \{i:\bm{\pi}_{i,K+1} =\bm{1}\}$.

\subsection{Measurement Budget}
Sequential thresholding satisfies the measurement budget in (\ref{eqn:budget}) provided $s$ grows sublinearly with $n$.  For brevity, we argue the procedure comes arbitrarily close to satisfying the measurement budget for large $n$:
\begin{eqnarray} \nonumber
  \E\left[\sum_{i,j} \pi_{i,j} \right]  &\leq & \sum_{k=0}^{K-1} \left((1-\rho)^k (n-s) \rho m + s \rho m \right) \\
  & \leq & m(n-s) +  msK \rho. \nonumber
\end{eqnarray}
Letting $K = \log n$, the procedure comes arbitrarily close to satisfying the constraint as $n$ grows large. To be rigorous in showing the procedure satisfies (\ref{eqn:budget}), $m$ can be replaced by $m-1$, and the analysis throughout holds.

\subsection{Ability of Sequential Thresholding}
We present the first of the three main theorems of the paper to quantify the performance of sequential thresholding.
\begin{thm}
\emph{Ability of sequential thresholding.} Provided
\begin{eqnarray} \label{eqn:SeqThesCond}
  m > \frac{\log s }{D(f_0||f_1)} + \frac{\log \log n}{D(f_0||f_1)}
\end{eqnarray}
sequential thresholding recovers $\S$ with high probability.
More precisely, if
\begin{eqnarray} \nonumber
  \lim_{n\rightarrow \infty} \frac{m}{\log \left(s \log n\right)} > \frac{1}{D(f_0||f_1)}
\end{eqnarray}
then $\P({\cal{E}}) \rightarrow 0$.
\end{thm}

\begin{proof}
From a union bound on the family wise error rate, we have
\begin{eqnarray} \label{eqn:fwer_st}
  \P({\cal{E}}) &\leq& (n-s) \alpha + s \beta.
\end{eqnarray}
Employing sequential thresholding, from the definition of $\gamma$, $\alpha = {(1-\rho)^K}$ and
\begin{eqnarray} \nonumber
\beta &=& \P \left(\bigcup_{k=1}^K  t_i^{(\rho m)} < \gamma  | i \in \S\right) \\
&\leq& K \P \left( t_i^{(\rho m)} < \gamma | i \in \S \right) \nonumber
\end{eqnarray}
where the inequality follows from a union bound.

We can further bound the false negative error event using the Chernoff-Stein Lemma \cite{Cover:1991:EIT:129837}, p. 384.  Consider a simple binary hypothesis test with a fixed probability of false positive at $\alpha_0 = 1-\rho$.  By the Chernoff-Stein Lemma, the false negative probability is then given as
\begin{eqnarray}  \nonumber
\P \left( t_i^{(\rho m)} < \gamma \mid i \in \S \right) \doteq e^{-\rho m D(f_0||f_1)}
\end{eqnarray}
where $a \doteq e^{-mD}$ is equivalent to
\begin{eqnarray} \nonumber
  \lim_{m\rightarrow \infty} \frac{1}{m} \log a = -D.
\end{eqnarray}
This implies, for any $\epsilon_1 > 0$, for sufficiently large $m$,
\begin{eqnarray}  \nonumber
\P \left( t_i^{(\rho m)} < \gamma \mid i \in \S \right) \leq e^{-\rho m (D(f_0||f_1) - \epsilon_1)}.
\end{eqnarray}
Letting $K = (1+\epsilon_2) \log n$, for sufficiently large $n$ and $m$, (\ref{eqn:fwer_st}) becomes
\begin{eqnarray} \nonumber
  \P({\cal{E}}) &\leq& \frac{(n-s)}{n^{(1+\epsilon_2)}}  + s (1+\epsilon_2) \log (n) \: e^{-\rho m (D(f_0||f_1) - \epsilon_1)}.
\end{eqnarray}
Hence, $\P({\cal{E}})$ goes to zero provided
\begin{eqnarray} \nonumber
  m \geq \frac{ \log( (1+\epsilon_2) s\log n) }{\rho (D(f_0||f_1)-\epsilon_1)}
\end{eqnarray}
which, as $\epsilon_1$ and $\epsilon_2$ can be made arbitrarily small, and $\rho$ can be made arbitrarily close to $1$, directly gives the theorem:
\begin{eqnarray} \nonumber
  m > \frac{\log( s\log n) }{D(f_0||f_1)}.
\end{eqnarray}
\end{proof}

\section{Lower Bound on Sequential Procedures}
In this section we derive a lower bound on the rate at which $m$ must grow with $n$ for any sequential procedure, and relate Sequential Thesholding to the high dimensional extension of the well known sequential probability ratio test (SPRT).   

\subsection{Limitation of any sequential procedure}
The lower bound for any sequential procedure is presented in the following theorem.
\begin{thm}
Consider any sequential measurement procedure.  Provided
\begin{eqnarray} \nonumber
  m < \frac{\log s}{D(f_0||f_1)}
\end{eqnarray}
the family wise error rate tends to one.  More precisely, if
\begin{eqnarray} \label{eqn:condreq}
  \lim_{n \rightarrow \infty} \frac{m}{\log s} < \frac{1}{D(f_0||f_1)}
\end{eqnarray}
then $\P({\cal{E}}) \rightarrow 1$.
\end{thm}
\begin{proof}
First, we show conditions under which the family wise error rate goes to one:
\begin{eqnarray} \nonumber
\P({\cal{E}}) &=& \P\left(\bigcup_{i \not \in \S} {\cal{E}}_i \cup \bigcup_{i \in \S} {\cal{E}}_i \right) \\ \nonumber
&=& 1 - \P\left(\bigcap_{i \not \in \S } {\cal{E}}_i^c \cap \bigcap_{i \in \S} {\cal{E}}_i^c\right) \\  \label{eqnPEe}
&=& 1 - (1-\beta)^s (1-\alpha)^{n-s} \nonumber \\  
&\geq& 1-e^{-\beta s} e^{-\alpha(n-s)}
\end{eqnarray}
which goes to one provided either
\begin{eqnarray}  \label{eqn:alpha_req}
  \alpha > \frac{1}{n-s}
\qquad \qquad
  \beta > \frac{1}{s}.
\end{eqnarray}

Second, for a simple binary hypothesis test, we can bound the expected number of measurements of \emph{any} sequential procedure with false positive and false negative probabilities $\alpha$ and $\beta$.  To simplify notation, define:
\begin{eqnarray} \nonumber
\E_0[N] =  \E\left[\sum_{j} \pi_{i,j}|i \not \in \S\right]  \quad \E_1[N] =  \E\left[\sum_{j} \pi_{i,j}|i \in \S\right]
\end{eqnarray}
that is, $\E_0[N]$ and $\E_1[N]$ are the expected number of measurements under $f_0$ and $f_1$ respectively.  From \cite{SeqAnalysis} p.21, we have
\begin{eqnarray} \nonumber
  E_0[N] \geq \frac{1}{D(f_0||f_1)}\left(\alpha \log \frac{\alpha}{1-\beta} + (1-\alpha) \log \frac{1-\alpha}{\beta}\right)
\end{eqnarray}
which is derived from a simple argument using Jensen's inequality.  The total expected number of measurements, constrained by the measurement budget, is
\begin{eqnarray} \label{eqn:constra}
   (n-s) E_0[N] + s E_1[N] = E\left[\sum_{i,j} \pi_{i,j}  \right] \leq mn
\end{eqnarray}

Dropping the $s E_1[N]$ term from (\ref{eqn:constra}), we need to find conditions under which the inequality 
\begin{eqnarray}  \nonumber
  \frac{n-s}{D(f_0||f_1)}\left(\alpha \log \frac{\alpha}{1-\beta} + (1-\alpha) \log \frac{1-\alpha}{\beta}\right) \leq mn 
\end{eqnarray}
implies $\P({\cal{E}}) \rightarrow 1$.
Dividing by $n\log s$, the inequality becomes
\begin{eqnarray} \nonumber
\frac{n-s}{D(f_0||f_1) n \log s} \left(\alpha \log \frac{\alpha}{1-\beta} + (1-\alpha) \log \frac{1-\alpha}{\beta}\right) \leq  \frac{m}{\log s}.
\end{eqnarray}
Imposing the condition in (\ref{eqn:condreq}) and cancelling $D(f_0||f_1)$ from both sides, the above inequality requires
\begin{eqnarray} \label{eqn:condre33}
    \lim_{n\rightarrow \infty} \frac{n-s}{ n \log s} \left(\alpha \log \frac{\alpha}{1-\beta} + (1-\alpha) \log \frac{1-\alpha}{\beta}\right) < 1 .
\end{eqnarray}
It is sufficient to show that (\ref{eqn:condre33}) implies either $\alpha > \frac{1}{n-s}$ or $\beta > \frac{1}{s}$ in the high dimensional limit.

With this in mind, let $\beta = \frac{1-\epsilon_1}{s}$, and $\alpha =\frac{1-\epsilon_2}{n-s}$ for some $\epsilon_1, \epsilon_2 \in [0,1)$.  Taking the limit as $n \rightarrow \infty$ in (\ref{eqn:condre33}) and reducing terms we have:
\begin{eqnarray}
\lim_{n \rightarrow \infty} \left( \cdot \right) = 1
\end{eqnarray}
which contradicts (\ref{eqn:condre33}), and negates our assumption that both $\beta = \frac{1-\epsilon_1}{s}$ and $\alpha =\frac{1-\epsilon_2}{n-s}$ for $\epsilon_1, \epsilon_2  \in [0,1)$.  Hence, by (\ref{eqn:alpha_req}), the family wise error rate must go to one, completing the proof.
\end{proof}

\subsection{The SPRT}
The sequential probability ratio test can be extended from simple binary hypothesis tests to the high dimensional case by simply considering $n$ parallel SPRTs.  Each individual SPRT operates by continuing to measure a component if the corresponding likelihood ratio is within an upper and lower boundary, and terminating measurement otherwise.  For scalars $A$ and $B$
\begin{eqnarray} \nonumber
  \pi_{i,j+1} &=& \left\{
                      \begin{array}{ll}
                        1 &\mathrm{if} \quad  \frac{A}{j} < t_i^{(j)}(y_{i,1},...,y_{i,j}) < \frac{B}{j} \\
                        0 &\mathrm{else}
                      \end{array}
                    \right. \\ \nonumber
\end{eqnarray}
where $t_i^{(j)}$ is the normalized log-likelihood ratio comprised of \emph{all} prior measurements (unlike sequential thresholding, in which the likelihood ratio is only formed using measurements from a single pass).  If $t_i^{(j)} < A/j$, the SPRT labels index $i$ as \emph{not} belonging to $\S$, and if $t_i^{(j)} > B/j$, index $i$ is assigned to $\S$.  For a thorough discussion of the SPRT, see \cite{SeqAnalysis}.

Sequential probability ratio tests are optimal for binary hypothesis tests in terms of minimum expected number of measurements for any error probabilities $\alpha$ and $\beta$ (shown originally in \cite{1948}), and this optimality can be translated to the high dimensional case.  
Consider a single component $i$, and the corresponding binary hypothesis test.  To be thorough, we restate the optimal property of the SPRT in the following lemma.
\begin{lem} \label{lemma11}
\emph{Optimality of the SPRT for simple binary tests \cite{chernoffSeq} (p.63).}  Consider an SPRT with expected number of measurements $\E_0[N]$ and $ \E_1[N]$, and corresponding error probabilities $\alpha$ and $\beta$.  Any other sequential test with expected number of measurements $\E_0[N]'$ and $\E_1[N]'$ and error probabilities $\alpha' \leq \alpha$ and $\beta' \leq \beta$ will also have $\E_0[N]' \geq \E_0[N]$ and  $\E_1[N]' \geq \E_1[N]$.
\end{lem}
\noindent In short, no procedure with the smaller error probabilities can have fewer measurements in expectation than the SPRT.  To translate the optimality of the SPRT to the high dimensional case, we introduce the following lemma.
\begin{lem} \label{lemma22}
\emph{Optimality of the SPRT.}
Consider $n$ component-wise sequential probability ratio tests used to estimate $\S$ each with error probabilities $\alpha$ and $\beta$, and with a total of $\E[ \sum_{i,j} \pi_{i,j} ]$ measurements in expectation. Any other component wise test with $\alpha'\leq \alpha$ and $\beta'\leq \beta$ will also have expected number of measurements $\E[ \sum_{i,j} \pi_{i,j} ]' \geq \E[ \sum_{i,j} \pi_{i,j} ]$.
\end{lem}
\begin{proof}
We can write the total expected number of measurements as:
\begin{eqnarray} \nonumber
  \E\left[ \sum_{i,j} \pi_{i,j} \right] = (n-s)\E_0[N]  + s \E_1[N]
\end{eqnarray}
which is monotonically increasing in both $\E_0[N]$ and $\E_1[N]$.  Together with \ref{lemma11}, this implies the lemma.
\end{proof}

\subsection{Comparison of the SPRT to Sequential Thresholding}
Although a fully rigorous proof is quite involved, using standard approximations for the sequential probability ratio test (again, see \cite{SeqAnalysis}) it is relatively straightforward to show the SPRT does achieve the lower bound presented above.

Sequential thresholding is similar in spirit to the SPRT.  In many scenarios, however, implementing the SPRT can be substantially more complicated, if not infeasible, when compared to sequential thresholding. To set the stopping boundaries, an SPRT requires knowledge of the underlying distributions as well as the level of sparsity $s$.  Even when these are available, only approximations relating error probabilities to the stopping boundaries can be derived in closed-form.

On the contrary, sequential thresholding does not require knowledge of $s$.  Since its sample requirements are within a factor a small factor of the lower bound, sequential thresholding is automatically {\em adaptive} to unknown levels of sparsity. Moreover, in practice, sequential thresholding needs only approximate knowledge of the distributions to operate (such that a substantial number of components that follow $f_0$ can be eliminated on each pass).

\section{Limitation of Non-Sequential Methods }
Our analysis would not be complete without comparison of sequential thresholding and the sequential lower bound to the performance limits of non-sequential methods.  To do so, we analyze performance of any non-sequential method using \emph{Chernoff Information}.
\begin{thm}
\emph{Limitation of non-sequential testing.}
Consider any non-sequential thresholding procedure. Provided
\begin{eqnarray} \label{eqn:asmpNS}
m < \frac{\log n}{D(f_1||f_0)}
\end{eqnarray}
the family wise error rate goes to 1.  To be precise, (\ref{eqn:asmpNS}) is equivalent to
\begin{eqnarray} \nonumber
  \lim_{n\rightarrow \infty} \frac{m}{\log n} < \frac{1}{D(f_1||f_0)}.
\end{eqnarray}
which implies $\lim_{n \rightarrow \infty} \P(\mathcal{E}) =1$. 
\end{thm}
\begin{proof}
From \cite{Cover:1991:EIT:129837}, p. 386, (Chernoff Information) and by (\ref{eqn:alpha_req}) any non-sequential test fails provided
\begin{eqnarray}
  \alpha \doteq e^{-m D(f_{\lambda}||f_{0})} > \frac{1}{n-s} \nonumber
\end{eqnarray}
or
\begin{eqnarray}
  \beta \doteq e^{-m D(f_{\lambda}||f_{1})} > \frac{1}{s} \nonumber
\end{eqnarray}
where
\begin{eqnarray} \nonumber
f_{\lambda} = \frac{f_0^\lambda f_1^{1-\lambda}}{\int_{\Omega} f_0^\lambda f_1^{1-\lambda}dy}
\end{eqnarray}
for $\lambda \in [0,1]$.
Hence, any sequential procedure fails provided
\begin{eqnarray} \nonumber
  m &<&  \min_{\lambda \in [0,1]} \max\left( \frac{\log(n-s)}{D(f_{\lambda}||f_{0})}, \frac{\log s}{D(f_{\lambda}||f_{1})}\right)
\end{eqnarray}
which is implied if
\begin{eqnarray} \nonumber
m &<& \frac{\log n }{D(f_{1}||f_{0})}
\end{eqnarray}
completing the proof.
\end{proof}

\section{Conclusion}
This paper showed sequential methods for support recovery of high dimensional sparse signals in noise can succeed using far fewer measurements than non-sequential methods.  Specifically, non-sequential methods require the number of measurements to grow logarithmically with the dimension, while sequential methods succeed if the number of measurements grows logarithmically with the level of sparsity. Additionally, a simple procedure termed sequential thresholding comes within a small additive factor of the lower bound in terms of number of measurements per dimension.

\bibliographystyle{IEEEtran}
\bibliography{Asilomar}

\end{document}